\documentclass{article}%
\usepackage{amsmath}
\usepackage{amsfonts}
\usepackage{amssymb}
\usepackage{graphicx}
\usepackage{float}
\usepackage{tikz}
\usepackage{tkz-graph}%
\providecommand{\U}[1]{\protect\rule{.1in}{.1in}}
\newtheorem{theorem}{Theorem}[section]

\newtheorem{definition}[theorem]{Definition}
\newtheorem{example}[theorem]{Example}

\newtheorem{lemma}[theorem]{Lemma}

\newtheorem{remark}[theorem]{Remark}

\newenvironment{proof}[1][Proof]{\noindent \textbf{#1.} }{\  \rule{0.5em}{0.5em}}
\begin{document}

\title{Simultaneously Moving Cops and Robbers }
\author{G. Konstantinidis and Ath. Kehagias}
\maketitle

\begin{abstract}
In this paper we study the \emph{concurrent} \emph{cops and robber
(CCCR)\ }game. CCCR follows the same rules as the classical, turn-based game,
except for the fact that the players move \emph{simultaneously}. The cops'
goal is to capture the robber and the \emph{concurrent cop number} of a graph
is defined the minimum number of cops which guarantees capture. For the
variant in which it it required to capture the robber in the \emph{shortest
possible time}, we let time to capture be the \emph{payoff function} of CCCR;
the (game theoretic) \emph{value} of CCCR is the optimal capture time and (cop
and robber) \emph{time optimal strategies} are the ones which achieve the
value. In this paper we prove the following.

\begin{enumerate}
\item For every graph $G$, the concurrent cop number is equal to the
\textquotedblleft classical\textquotedblright\ cop number.

\item For every graph $G$, CCCR has a value, the cops have an optimal strategy
and, for every $\varepsilon>0$, the robber has an $\varepsilon$-optimal strategy.

\end{enumerate}
\end{abstract}

\section{Introduction\label{sec01}}

In this paper we study the \emph{concurrent} \emph{cops and robber
(CCCR)\ }game. In the classical CR\ game
\cite{nowakowski1983vertex,quilliot1978jeux} each player observes the other
player's move before he performs his own. On the other hand, in concurrent
CR\ the players move \emph{simultaneously}. In all other aspects, the
concurrent game (henceforth CCCR) follows the same rules as the classical,
turn-based game (henceforth TBCR).

The CCCR game (similarly to TBCR)\ can be considered as either a \emph{game of
kind} (the cops' goal is to capture the robber) or a \emph{game of degree}
(the cops' goal is to capture the robber \emph{in the shortest possible
time})\footnote{This terminology is due to Isaacs
\cite{isaacs1999differential}.}.

This paper is organized as follows. In Section \ref{sec02} we define
preliminary concepts and notation and use these to define the CCCR\ game
rigorously. In Section \ref{sec03} we concentrate on the \textquotedblleft game of kind"
aspect:\ we define the \emph{concurrent cop number} $\widetilde{c}\left(
G\right)  $ and prove that, for every graph $G$, it is equal to the
\textquotedblleft classical\textquotedblright\ cop number $c\left(  G\right)
$. In Section \ref{sec04} we concentrate on the \textquotedblleft game of degree" aspect:\ we
equip CCCR with a \emph{payoff function} (namely the time required to capture
the robber) and prove that (a)\ CCCR\ has a game theoretic \emph{value}, (b)
the cops have an optimal strategy and (c) for every $\varepsilon>0$ the robber
has an $\varepsilon$-optimal strategy; in addition we provide an algorithm for
the computation of the value and the optimal strategies. In Section
\ref{sec05} we discuss related work. Finally, in Section \ref{sec06} we
present our conclusions and future research directions.

\section{Preliminaries\label{sec02}}

In this section, as well as in the rest of the paper, we will mainly concern
ourselves with the case of a single cop; this is reflected in the following
definitions and notation. In case $K>1$ cops are considered, this will be
stated explicitly; the extension of definitions and notation is straightforward.

\subsection{Definition of the CCCR Game}

Both CCCR\ and TBCR are played on an undirected, simple and connected graph
$G=\left(  V,E\right)  $ by two players called $C$ and $R$. Player $C$,
controlling $K$ cops (with $K\geq1$) pursues a single robber controlled by
player $R$ (we will sometimes call both the cops and robber \emph{tokens}). We
assume the reader is familiar with the rules of TBCR and proceed to present
the rules of CCCR for the case of $K=1$ (a single cop).

\begin{enumerate}
\item The game starts from given initial positions: the cop is located at
$x_{0}\in V$ and the robber at $y_{0}\in V$.

\item At the $t$-th round ($t\in\mathbb{N}$) $C$ moves the cop to $x_{t}\in
N\left[  x_{t-1}\right]  $ and simultaneously $R$ moves the robber to
$y_{t}\in N\left[  y_{t-1}\right]  $\footnote{$N\left[  u\right]  $ denotes
the \emph{closed neighborhood} of node $u$, i.e., the set containing $u$
itself and all nodes connected to $u$ by an edge.}$.$

\item At every round both players know the current cop and robber location
(and remember all past locations).

\item A \emph{capture} occurs at the smallest $t\in\mathbb{N}$ for which
either of the following conditions holds:

\begin{enumerate}
\item The cop is located at $x_{t}$, the robber is located at $y_{t}$, and
$x_{t}=y_{t}$. This capture condition is the same as in TBCR.

\item The cop is located at $x_{t-1}$ and moves to $y_{t-1}$, while the robber
is located at $y_{t-1}$ and moves to $x_{t-1}$. We will call this
\textquotedblleft\emph{en passant}\textquotedblright\ capture; it does not
have an analog in TBCR.
\end{enumerate}

\item $C$ wins if capture takes place for some $t\in\mathbb{N}$. Otherwise,
$R$ wins. The game analysis becomes easier if we assume that the game always
lasts an infinite number of rounds; if a capture occurs at $t_{c}$, then we
will have $x_{t}=y_{t}=x_{t_{c}}$ for all $t\geq t_{c}$.
\end{enumerate}

We will denote the above defined game, played on graph $G=\left(  V,E\right)
$ and starting from initial position $\left(  x,y\right)  \in V^{2}$ by
$\Gamma_{\left(  x,y\right)  }^{G}$. In case the game is played with $K$ cops,
it will be denoted by $\Gamma_{\left(  x,y\right)  }^{G,K}$ (in this case
$x\in V^{K}$).

\subsection{Nomenclature and Notation}

The following quantities will be used in the subsequent analysis (once again,
we present definitions for the case of $K=1$). Some of them require two
separate definitions:\ one for TBCR and another for CCCR.

\begin{definition}
\label{prp0201}A \emph{position} in TBCR is a triple $\left(  x,y,P\right)  $
where $x\in V$ is the cop location, $y\in V$ is the robber location and
$P\in\left\{  C,R\right\}  $ is the player whose turn it is to move. We also
have $\left\vert V\right\vert +1$ additional positions:

\begin{enumerate}
\item the position $\left(  \emptyset,\emptyset,C\right)  $ corresponds to the
beginning of the game, before either player has placed his token;

\item the positions $\left(  x,\emptyset,R\right)  $, $x\in V$, correspond to
the phase of the game in which C has placed the cop but R has not placed the robber.
\end{enumerate}

\noindent The set of all TBCR\ positions is denoted by $S=V\times
V\times\left\{  C,R\right\}  $.
\end{definition}

\begin{definition}
\label{prp0202}A \emph{position} in CCCR is a pair $\left(  \widetilde
{x},\widetilde{y}\right)  $ where $\widetilde{x}\in V$ is the cop location and
$\widetilde{y}\in V$ is the robber location. The set of all CCCR positions is
denoted by $\widetilde{S}=V\times V$.
\end{definition}

\begin{definition}
\label{prp0203}A \emph{history} is a position sequence of finite or infinite
length. The set of all game histories \ of \emph{any} finite length is denoted
by $S^{\ast}$ for TBCR\ and $\widetilde{S}^{\ast}$ for CCCR. The set of all
infinite game histories is denoted by $S^{\infty}$ for TBCR\ and
$\widetilde{S}^{\infty}$ for CCCR.
\end{definition}

In both TBCR and CCCR, the players' \emph{moves} are graph nodes, e.g.,
$x,y\in V$. Given the next move (in TBCR) or moves (in CCCR) the next game
position is determined by the \emph{transition function}, which encodes the
rules of the respective game.

\begin{definition}
\label{prp0204}In TBCR, the \emph{transition function} $Q:S\times V\rightarrow
S$ is defined as follows:%
\[%
\begin{array}
[c]{ll}%
\text{when }x=y: & Q\left(  \left(  x,y,C\right)  ,x^{\prime}\right)  =\left(
x,x,R\right) \\
\text{when }x\neq y\text{ and }x^{\prime}\in N\left[  x\right]  : & Q\left(
\left(  x,y,C\right)  ,x^{\prime}\right)  =\left(  x^{\prime},y,R\right) \\
\text{when }x\neq y\text{ and }x^{\prime}\notin N\left[  x\right]  : &
Q\left(  \left(  x,y,C\right)  ,x^{\prime}\right)  =\left(  x,y,R\right) \\
\text{when }x=y: & Q\left(  \left(  x,y,R\right)  ,y^{\prime}\right)  =\left(
x,x,C\right) \\
\text{when }x\neq y\text{ and }y^{\prime}\in N\left[  y\right]  : & Q\left(
\left(  x,y,R\right)  ,y^{\prime}\right)  =\left(  x,y^{\prime},C\right) \\
\text{when }x\neq y\text{ and }y^{\prime}\notin N\left[  y\right]  : &
Q\left(  \left(  x,y,R\right)  ,y^{\prime}\right)  =\left(  x,y,C\right)
\end{array}
\]

\end{definition}

\begin{definition}
\label{prp0205}In CCCR, the \emph{transition function} $\widetilde
{Q}:\widetilde{S}\times V\times V\rightarrow\widetilde{S}$ is defined as
follows:%
\[%
\begin{array}
[c]{ll}%
\text{when }x=y: & \widetilde{Q}\left(  \left(  x,y\right)  ,x^{\prime
},y^{\prime}\right)  =\left(  x,x\right) \\
\text{when }x\neq y\text{ and }x^{\prime}\in N\left[  x\right]  \text{ and
}y^{\prime}\in N\left[  y\right]  & \\
\qquad\qquad\text{if }x^{\prime}=y\text{ and }y^{\prime}=x: & \widetilde
{Q}\left(  \left(  x,y\right)  ,x^{\prime},y^{\prime}\right)  =\left(
x^{\prime},x^{\prime}\right) \\
\qquad\qquad\text{otherwise}: & \widetilde{Q}\left(  \left(  x,y\right)
,x^{\prime},y^{\prime}\right)  =\left(  x^{\prime},y^{\prime}\right) \\
\text{when }x\neq y\text{ and }x^{\prime}\notin N\left[  x\right]  \text{ and
}y^{\prime}\in N\left[  y\right]  : & \widetilde{Q}\left(  \left(  x,y\right)
,x^{\prime},y^{\prime}\right)  =\left(  x,y^{\prime}\right) \\
\text{when }x\neq y\text{ and }x^{\prime}\in N\left[  x\right]  \text{ and
}y^{\prime}\notin N\left[  y\right]  : & \widetilde{Q}\left(  \left(
x,y\right)  ,x^{\prime},y^{\prime}\right)  =\left(  x^{\prime},y\right) \\
\text{when }x\neq y\text{ and }x^{\prime}\notin N\left[  x\right]  \text{ and
}y^{\prime}\notin N\left[  y\right]  : & \widetilde{Q}\left(  \left(
x,y\right)  ,x^{\prime},y^{\prime}\right)  =\left(  x,y\right)
\end{array}
\]

\end{definition}

\noindent The above rules have the following consequences (which will
facilitate our subsequent analysis).

\begin{enumerate}
\item The CR\ game continues for an infinite number of rounds; but if a
capture occurs at some time $t_{c}$, the cop and robber locations remain fixed
for all subsequent times.

\item The transition function accepts \textquotedblleft
illegal\textquotedblright\ moves (e.g., $x^{\prime}\notin N\left[  x\right]
$) as input but \textquotedblleft ignores" them, in the sense that they have no influence on
the location of tokens.
\end{enumerate}

Roughly speaking, a \emph{strategy} is a rule which, given a game history,
prescribes a player's next move. In CCCR the players gain an advantage by
using \emph{randomized} or \emph{mixed strategies}.

\begin{definition}
\label{prp0206}A\ \emph{randomized }or \emph{mixed strategy} is a function
$\widetilde{\pi}:\widetilde{S}^{\ast}\times V\rightarrow\left[  0,1\right]  $,
which satisfies%
\[
\forall\left(  \left(  \widetilde{x}_{0},\widetilde{y}_{0}\right)  ,\left(
\widetilde{x}_{1},\widetilde{y}_{1}\right)  ,...,\left(  \widetilde{x}%
_{t},\widetilde{y}_{t}\right)  \right)  \in\widetilde{S}^{\ast}:\sum
_{\widetilde{z}\in V}\widetilde{\pi}\left(  \widetilde{z}|\left(
\widetilde{x}_{0},\widetilde{y}_{0}\right)  ,\left(  \widetilde{x}%
_{1},\widetilde{y}_{1}\right)  ,...,\left(  \widetilde{x}_{t},\widetilde
{y}_{t}\right)  \right)  =1
\]
and gives the probability that at time $t$ the player moves into node
$\widetilde{z}$, given that the game has started at position $\left(
\widetilde{x}_{0},\widetilde{y}_{0}\right)  $ and progressed through positions
$\left(  \widetilde{x}_{1},\widetilde{y}_{1}\right)  ,...,\left(
\widetilde{x}_{t},\widetilde{y}_{t}\right)  $.
\end{definition}

\noindent Two classes of strategies will be of special interest to us.

\begin{definition}
\label{prp0207}A strategy $\widetilde{\pi}$ is called \emph{memoryless} iff
\[
\forall\left(  \left(  \widetilde{x}_{0},\widetilde{y}_{0}\right)
,...,\left(  \widetilde{x}_{t},\widetilde{y}_{t}\right)  \right)
\in\widetilde{S}^{\ast},\forall\widetilde{z}\in V:\widetilde{\pi}\left(
\widetilde{z}|\left(  \widetilde{x}_{0},\widetilde{y}_{0}\right)  ,...,\left(
\widetilde{x}_{t},\widetilde{y}_{t}\right)  \right)  =\widetilde{\pi}\left(
\widetilde{z}|\left(  \widetilde{x}_{t},\widetilde{y}_{t}\right)  \right)  ,
\]
i.e., the player's move depends only on the \emph{current} game position.
\end{definition}

\begin{definition}
\label{prp0208}A strategy $\widetilde{\pi}$ is called \emph{deterministic} iff%
\[
\forall\left(  \left(  \widetilde{x}_{0},\widetilde{y}_{0}\right)
,...,\left(  \widetilde{x}_{t},\widetilde{y}_{t}\right)  \right)
\in\widetilde{S}:\exists\widetilde{z}:\widetilde{\pi}\left(  \widetilde
{z}|\left(  \widetilde{x}_{0},\widetilde{y}_{0}\right)  ,...,\left(
\widetilde{x}_{t},\widetilde{y}_{t}\right)  \right)  =1,
\]
i.e., for every game history $\left(  \widetilde{x}_{0},\widetilde{y}%
_{0}\right)  ,...,\left(  \widetilde{x}_{t},\widetilde{y}_{t}\right)  $, there
is a position $\widetilde{z}$ to which the player will move with certainty.
\end{definition}

\noindent If $\widetilde{\pi}$ is deterministic, it can be equivalently
described by a function $\widetilde{\sigma}:\widetilde{S}^{\ast}\rightarrow V$
which is determined by $\widetilde{\pi}$ as follows:%
\[
\widetilde{\sigma}\left(  \left(  \widetilde{x}_{0},\widetilde{y}_{0}\right)
,...,\left(  \widetilde{x}_{t},\widetilde{y}_{t}\right)  \right)
=\widetilde{z}\text{ iff }\widetilde{\pi}\left(  \widetilde{z}|\left(
\widetilde{x}_{0},\widetilde{y}_{0}\right)  ,...,\left(  \widetilde{x}%
_{t},\widetilde{y}_{t}\right)  \right)  =1.
\]
Similarly, if $\widetilde{\pi}$ is memoryless and deterministic, it can be
equivalently described by a function $\widetilde{\sigma}:\widetilde
{S}\rightarrow V$ which is determined by $\widetilde{\pi}$ as follows:%
\[
\widetilde{\sigma}\left(  \widetilde{x}_{t},\widetilde{y}_{t}\right)  =z\text{
iff }\widetilde{\pi}\left(  \widetilde{z}|\left(  \widetilde{x}_{t}%
,\widetilde{y}_{t}\right)  \right)  =1.
\]

The above definitions and remarks concern CCCR strategies. Regarding TBCR
strategies, it is well known \cite{hahn2006game} that both players lose
nothing by restricting themselves to memoryless deterministic strategies of
the form $\sigma:S\rightarrow V$. In other words, if player $P$ uses the
strategy $\sigma$ and the current game position is $\left(  x,y,P\right)  $
(which means that it is $P$'s turn to move) P moves his token into node
$\sigma\left(  x,y,P\right)  $. Obviously $P$ will only use $\sigma$ when it
is his turn to move; hence we can use the notation $\sigma_{C}\left(
x,y\right)  $ when talking about \emph{cop strategies} and $\sigma_{R}\left(
x,y\right)  $ when talking about \emph{robber strategies}. Note that a cop
strategy is also defined for the initial position $\left(  \emptyset
,\emptyset,C\right)  $ and a robber strategy is also defined for the initial
positions $\left(  x,\emptyset,R\right)  $ (for every $x\in V$).

\section{Cop Numbers\label{sec03}}

In the \textquotedblleft classical\textquotedblright\ TBCR\ game we have the following.

\begin{definition}
\label{prp0301}The cop number $c\left(  G\right)  $ of a graph $G$ is the
minimum number of cops sufficient to capture the robber when TBCR\ is played
(optimally by both players) on $G$.
\end{definition}

Note that in the above definition optimal play includes optimal initial
placement (in the 0-th round) of the cops and robber on $G$. On the other
hand, in the CCCR\ game $\Gamma_{\left(  x_{0},y_{0}\right)  }^{G,K}$ the
initial cops and robber positions are given (rather than chosen by the
players). A reasonable definition of cop number should account for all
possible initial positions. Hence we have the following.

\begin{definition}
\label{prp0302}The concurrent cop number $\widetilde{c}\left(  G\right)  $ of
graph $G$ is the minimum number of cops sufficient to ensure capture with
probability one for every initial position $\left(  \widetilde{x}%
_{0},\widetilde{y}_{0}\right)  $, when CCCR\ is played (optimally by both
players) on $G$.
\end{definition}

Note also the expression \textquotedblleft capture with probability
one\textquotedblright\ in Definition \ref{prp0302}. This is different from
\textquotedblleft certain capture\textquotedblright\ in the sense that there
may exist infinite game histories in which capture does not occur, but the
probability of any such infinite history materializing is zero\footnote{This
point is further discussed in Section \ref{sec06}.}.

In what follows, whenever we mention an arbitrary robber (or cop)\ move
sequence $y_{0},y_{1},y_{2},...$ we assume that it is a \emph{legal} move
sequence, i.e., for all $t$ we have $y_{t+1}\in N\left[  y_{t}\right]  $.
Also, if capture occurs at time $t_{c}$, the robber's (and cop's)\ location
remains fixed at $y_{t}=y_{t_{c}}$, irrespective of the moves $y_{t_{c}%
+1},y_{t_{c}+2},...$ .

\begin{lemma}
\label{prp0303}$c\left(  G\right)  =1\Rightarrow\widetilde{c}\left(  G\right)
=1$.
\end{lemma}

\begin{proof}
We select an arbitrary graph $G$ with $c\left(  G\right)  =1$ and fix it for
the rest of the proof. Both TBCR and CCCR will be played on this $G$. We let
$n=\left\vert V\right\vert $, i.e., $n$ is the number of nodes of $G$.

We will prove the proposition by constructing a (deterministic and memoryless)
cop strategy $\widetilde{\pi}_{C}^{\#}$ which guarantees, for every starting
position, CCCR capture with probability $1$.\ 

An essential component of $\widetilde{\pi}_{C}^{\#}$ is a deterministic cop
strategy $\widetilde{\sigma}_{C}^{\ast}$, constructed from another
deterministic, memoryless cop strategy $\sigma_{C}^{\ast}$ which guarantees
capture in the TBCR\ game. Since $c\left(  G\right)  =1$ we know
\cite{hahn2006game} that such a $\sigma_{C}^{\ast}$ exists and guarantees
capture in at most $T$ rounds, where $T$ depends only on $G$. Furthermore
recall that we have defined TBCR so that after capture takes place both C\ and
R stay in place. The rest of the proof will be divided in two parts.

\noindent\underline{\textbf{Part 1}}. Consider the CCCR game and assume that,
for every time $t$, $C$ knows $R$'s next move (this assumption will be removed
in Part 2). Take an arbitrary starting position $\widetilde{s}_{0}=\left(
\widetilde{x}_{0},\widetilde{y}_{0}\right)  $ and suppose that at time $t$,
when the position is $\left(  \widetilde{x}_{t},\widetilde{y}_{t}\right)  $,
$C$ (knowing that $R$'s next move will be $\widetilde{y}_{t+1}$) plays
$\widetilde{x}_{t+1}=\widetilde{\sigma}_{C}^{\ast}\left(  \widetilde{x}%
_{t},\widetilde{y}_{t+1}\right)  =\sigma_{C}^{\ast}\left(  \widetilde{x}%
_{t},\widetilde{y}_{t+1}\right)  $.\footnote{Note that $\widetilde{\sigma}%
_{C}^{\ast}$ is deterministic and only uses two inputs: one is $\widetilde
{x}_{t}$ (from the previous round) ands the other is $\widetilde{y}_{t+1}$
(from the current round). Hence $\widetilde{\sigma}_{C}^{\ast}$ is memoryless
in the sense that it only requires knowledge of the immediate past position,
but it is also \emph{prescient} in the sense that it requires knowledge of the
current robber move.} Then, for any robber moves $\widetilde{y}_{1}%
,\widetilde{y}_{2},...$ in rounds $t=1,2,...$ the sequence of game positions
will be:
\[
\left(  \widetilde{x}_{0},\widetilde{y}_{0}\right)  ,\left(  \widetilde{x}%
_{1}=\sigma_{C}^{\ast}\left(  \widetilde{x}_{0},\widetilde{y}_{1}\right)
,\widetilde{y}_{1}\right)  ,\left(  \widetilde{x}_{2}=\sigma_{C}^{\ast}\left(
\widetilde{x}_{1},\widetilde{y}_{2}\right)  ,\widetilde{y}_{2}\right)
,...,\left(  \widetilde{x}_{t}=\sigma_{C}^{\ast}\left(  \widetilde{x}%
_{t-1},\widetilde{y}_{t}\right)  ,\widetilde{y}_{t}\right)  ,...
\]
We will prove that $\widetilde{x}_{T}=\widetilde{y}_{T}$, i.e., capture
results in at most $T$ rounds and this will be true with certainty for any
starting position $\widetilde{s}_{0}=\left(  \widetilde{x}_{0},\widetilde
{y}_{0}\right)  $ and robber moves $\widetilde{y}_{1},\widetilde{y}_{2},...$ thereafter.

To show this, consider a TBCR game in which, at the end of the initial round
($t=0$) the position is $\left(  x_{0},y_{0},C\right)  =\left(  \widetilde
{x}_{0},\widetilde{y}_{1},C\right)  $. Further suppose that $C$ uses
$\sigma_{C}^{\ast}$ and $R$ plays the moves $y_{1},y_{2},...$ with
$y_{t}=\widetilde{y}_{t+1}$ (for $t=0,1,...$). Note that, given $y_{0}%
=\widetilde{y}_{1}$, and also that $\widetilde{y}_{1},\widetilde{y}_{2},...$
are legal robber moves in CCCR, the resulting robber moves $y_{1},y_{2},...$in
TBCR are also legal. Moreover recall that, for any given starting position
$(x_{0},y_{0})$, when the robber moves are $y_{1},y_{2},...$ and $C$ uses
$\sigma_{C}^{\ast}$, we get a sequence of cop and robber locations of the
following form:%
\[
x_{0},y_{0},x_{1}=\sigma_{C}^{\ast}\left(  x_{0},y_{0}\right)  ,y_{1}%
,x_{2}=\sigma_{C}^{\ast}\left(  x_{1},y_{1}\right)  ,y_{2},...,x_{t}%
=\sigma_{C}^{\ast}\left(  x_{t-1},y_{t-1}\right)  ,y_{t},...
\]
Given $y_{t}=\widetilde{y}_{t+1}$ for $t=0,1,...,$ $x_{0}=\widetilde{x}_{0}$
and that C uses $\sigma_{C}^{\ast}$, we get $x_{1}=\sigma_{C}^{\ast}\left(
x_{0},y_{0}\right)  =\sigma_{C}^{\ast}\left(  \widetilde{x}_{0},\widetilde
{y}_{1}\right)  =\widetilde{x}_{1}$, $x_{2}=\sigma_{C}^{\ast}\left(
x_{1},y_{1}\right)  =\sigma_{C}^{\ast}\left(  \widetilde{x}_{1},\widetilde
{y}_{2}\right)  =\widetilde{x}_{2}$, $...$, $x_{t}=\sigma_{C}^{\ast}\left(
x_{t-1},y_{t-1}\right)  =\sigma_{C}^{\ast}\left(  \widetilde{x}_{t-1}%
,\widetilde{y}_{t}\right)  =\widetilde{x}_{t}$, $...$ . Thus the resulting
sequence of cop and robber locations in TBCR is:%
\begin{equation}
\widetilde{x}_{0},\widetilde{y}_{1},\widetilde{x}_{1},\widetilde{y}%
_{2},\widetilde{x}_{2},\widetilde{y}_{3},...,\widetilde{x}_{t},\widetilde
{y}_{t+1},... \label{eq002}%
\end{equation}
Since $\sigma_{C}^{\ast}$ guarantees capture by time $T$ in TBCR, we have
$x_{T}=y_{T}$, irrespective of the moves $y_{1},y_{2},...$ . In fact we will
have $x_{T}=y_{T-1}$, i.e., at most in the first (i.e. cop) phase of round
$T$, $C$ captures $R$, or else (i.e., if $x_{T}\neq y_{T-1}$) $R$ can stay put
in this round and then $x_{T}\neq y_{T}$, which is a contradiction. Since
$x_{T}=\widetilde{x}_{T}$ and $y_{T-1}$ $=\widetilde{y}_{T}$ from
$x_{T}=y_{T-1}$ we have $\widetilde{x}_{T}=\widetilde{y}_{T}$.

We conclude that also in the CCCR\ game, for any starting position
$\widetilde{s}_{0}=(\widetilde{x}_{0},\widetilde{y}_{0})$ and subsequent
robber moves $\widetilde{y}_{1},\widetilde{y}_{2},...$, , capture takes place
by the $T$-th round at the latest. We repeat that this holds under the
assumption that: in each round $t$, $C$ knows $R$'s next move $\widetilde
{y}_{t+1}$.

\noindent\underline{\textbf{Part 2}}. In the actual CCCR game $C$\ will not
know $R$'s next move $\widetilde{y}_{t+1}$; however he can always \emph{guess}
$\widetilde{y}_{t+1}$ to be $v$. Suppose that, when $R$ is at $\widetilde
{y}_{t}$, $C$\ guesses with uniform probability $\frac{1}{\left\vert N\left[
\widetilde{y}_{t}\right]  \right\vert }$ that $R$ will move to $v\in N\left[
\widetilde{y}_{t}\right]  $. Let $\widetilde{y}_{t+1}$ be $R$'s actual move at
$t+1$ and $\widehat{y}_{t+1}$ be $C$'s guess of that move. We have
\[
\Pr\left(  \widehat{y}_{t+1}=\text{ }v|\widetilde{y}_{t+1}=v\right)  =\frac
{1}{\left\vert N\left[  \widetilde{y}_{t}\right]  \right\vert }\geq\frac{1}{n}%
\]
and
\begin{align*}
&  \Pr\left(  C\text{ guesses }R\text{'s move correctly}\right) \\
&  =\Pr\left(  \widehat{y}_{t+1}=\widetilde{y}_{t+1}\right)  =\sum_{v\in
N\left[  \widetilde{y}_{t}\right]  }\Pr\left(  \widehat{y}_{t+1}=\text{
}v|\widetilde{y}_{t+1}=v\right)  \Pr\left(  \widetilde{y}_{t+1}=v\right)
\geq\frac{1}{n}\sum_{v\in N\left[  \widetilde{y}_{t}\right]  }\Pr\left(
\widetilde{y}_{t+1}=v\right)  =\frac{1}{n}.
\end{align*}
In other words, $C$ guesses correctly $R$'s next move with probability at
least $\frac{1}{n}$. It follows that $C$ guesses correctly $R$'s next $T$
moves (and captures $R$)\ with probability at least $\left(  \frac{1}%
{n}\right)  ^{T}$.

Now we define the following set of CCCR infinite game histories:%
\begin{align*}
\forall k\in\mathbb{N}:A_{k}  &  =\left\{  \mathbf{s}:\mathbf{s}\in
\widetilde{S}^{\infty}\text{ and }R\text{ is still free after the first
}k\cdot T\text{ rounds}\right\}  \text{,}\\
A  &  =\lim\sup A_{k}=\cap_{m=1}^{\infty}\cup_{k=m}^{\infty}A_{k}.
\end{align*}
Since $A_{k+1}\subseteq A_{k}$ (for all $k\in\mathbb{N}$)\ we have
\[
A=\cap_{m=1}^{\infty}\cup_{k=m}^{\infty}A_{k}=\cap_{m=1}^{\infty}%
A_{m}=\left\{  \mathbf{s}:\text{ }\mathbf{s}\in\widetilde{S}^{\infty}\text{
and }\forall m\in\mathbb{N}:\text{ }R\text{ is still free after the first
}m\cdot T\text{ rounds}\right\}  .
\]
In other words, $A$ is the set of all CCCR\ infinite game histories in which
$R$ is never captured. Since
\[
\sum_{k=1}^{\infty}\Pr\left(  A_{k}\right)  \leq\sum_{k=1}^{\infty}\left(
1-\left(  \frac{1}{n}\right)  ^{T}\right)  ^{k}<\infty
\]
we have (from the first Borel-Cantelli lemma \cite{billingsley2008probability}%
)\ that $\Pr\left(  A\right)  =0$.

To sum up:\ using the deterministic memoryless strategy $\sigma_{C}^{\ast}$ in
conjunction with uniform guessing, we obtain a randomized memoryless strategy
$\widetilde{\pi}_{C}^{\#}$ which guarantees capture with probability one in
the CCCR\ game; thus $\widetilde{c}\left(  G\right)  =1$ as claimed.
\end{proof}

In case of a graph $G$ with $c(G)=K$, it is straightforward to extend the
previous argument using $K$ cops. Note however that, at the end of the proof
we will know that:\ if $K$ cops are \emph{required} to capture the robber in
TBCR, then $K$ cops \emph{suffice} to capture (with probability one)\ the
robber in CCCR. Hence we have the following.

\begin{lemma}
\label{prp0304}$c\left(  G\right)  =K\Rightarrow\widetilde{c}\left(  G\right)
\leq K$.
\end{lemma}

Next we show the \textquotedblleft reverse\textquotedblright\ of Lemma
\ref{prp0301}.

\begin{lemma}
\label{prp0305}$\widetilde{c}\left(  G\right)  =1\Rightarrow c\left(
G\right)  =1$.
\end{lemma}

\begin{proof}
We select an arbitrary graph $G$ with $\widetilde{c}\left(  G\right)  =1$ and
fix it for the rest of the proof. Both TBCR and CCCR will be played on this
$G$. The Lemma can be stated equivalently as:
\[
c\left(  G\right)  >1\Rightarrow\widetilde{c}\left(  G\right)  >1
\]
and this is what we will prove.

If $c\left(  G\right)  >1$ then there exists a (memoryless and
deterministic)\ winning robber strategy $\sigma_{R}^{\ast}$ for TBCR with one
cop on $G$. More specifically, $\sigma_{R}^{\ast}$ guarantees that, for every
cop starting position $x_{0}$, the robber will never be captured.

Choose any $\widetilde{x}_{0}\in V$ and let $\widetilde{y}_{0}=\sigma
_{R}^{\ast}$ $\left(  \widetilde{x}_{0},\emptyset\right)  $. Using $\sigma
_{R}^{\ast}$, we will construct a CCCR robber strategy $\widetilde{\sigma}%
_{R}^{\ast}$ such that: when CCCR (played on $G$ with a single cop) starts
from position $\left(  \widetilde{x}_{0},\widetilde{y}_{0}\right)  $ and R
uses $\widetilde{\sigma}_{R}^{\ast}$, the capture probability is zero. This,
clearly, implies that $\widetilde{c}\left(  G\right)  >1$.

It suffices to define $\widetilde{\sigma}_{R}^{\ast}$ only for the case when
CCCR\ starts from $\left(  \widetilde{x}_{0},\widetilde{y}_{0}\right)  $, as follows.

\begin{enumerate}
\item In round $t=1$: $\widetilde{y}_{1}=\widetilde{y}_{0}$ ($R$ stays put);

\item In rounds $t=2,3,...$, $R$ plays according to $\sigma_{R}^{\ast}$ . In
other words, if $\widetilde{x}_{t-1}=u$ and $\widetilde{y}_{t-1}=v$, then
$\widetilde{y}_{t}=\widetilde{\sigma}_{R}^{\ast}(u,v)=\sigma_{R}^{\ast}(u,v)$.
\end{enumerate}

Clearly, $\widetilde{\sigma}_{R}^{\ast}$ is not strictly memoryless. The move
$\widetilde{y}_{1}=\widetilde{y}_{0}$ depends not only on the game position
$\left(  \widetilde{x}_{0},\widetilde{y}_{0}\right)  $ but also on the fact
that this is the first round. However, the part of $\widetilde{\sigma}%
_{R}^{\ast}$ used in rounds $t\geq2$ is memoryless.

Suppose that in CCCR (starting from $\left(  \widetilde{x}_{0},\widetilde
{y}_{0}\right)  $)$\ R$ plays the strategy $\widetilde{\sigma}_{R}^{\ast}$
while $C$\ plays any move sequence $\widetilde{x}_{1},\widetilde{x}_{2},...$ .
To prove that capture will never occur, consider a TBCR in which $R$ plays the
strategy $\sigma_{R}^{\ast}$ and $C$\ plays the same move sequence
$\widetilde{x}_{0},\widetilde{x}_{1},...$ as in CCCR. Since $\sigma_{R}^{\ast
}$ is winning, capture will never take place in TBCR; as will be shown, this
implies capture will never occur in CCCR either and, since this holds for
\emph{any }$\widetilde{x}_{1},\widetilde{x}_{2},...$, we will conclude that
$\widetilde{c}\left(  G\right)  >1$.

Let $y_{0},y_{1},y_{2}...$ be the robber moves occurring in TBCR, given that
$R$\ plays $\sigma_{R}^{\ast}$ and $C$\ plays $\widetilde{x}_{0},\widetilde
{x}_{1},...$. Let us use $d\left(  u,v\right)  $ to denote the distance of
nodes $u,v$ in $G$, i.e., the length of shortest path between $u$ and $v$.
Obviously we have%
\begin{equation}
\forall t\geq0:d\left(  \widetilde{x}_{t+1},y_{t}\right)  \geq1 \label{eq0201}%
\end{equation}
(if we had $d\left(  \widetilde{x}_{t+1},y_{t}\right)  =0$ then $\sigma
_{R}^{\ast}$ would not be a winning strategy). Furthermore%
\begin{equation}
\forall t\geq0:\widetilde{y}_{t+1}=y_{t}. \label{eq0202}%
\end{equation}
Indeed, $\widetilde{y}_{1}=y_{0}$ by construction and if, for some $n$, we
have $\widetilde{y}_{n}=y_{n-1}$, then%
\[
\widetilde{y}_{n+1}=\widetilde{\sigma}_{R}^{\ast}(\widetilde{x}_{n}%
,\widetilde{y}_{n})=\sigma_{R}^{\ast}(\widetilde{x}_{n},y_{n-1})=y_{n}.
\]
From (\ref{eq0201}) and (\ref{eq0202})\ follows that%
\begin{equation}
\forall t:1\leq d\left(  \widetilde{x}_{t},\widetilde{y}_{t}\right)  .
\label{eq0203}%
\end{equation}
This almost completes the proof that capture never occurs in CCCR. However, we
must also consider the possibility of an \textquotedblleft en
passant\textquotedblright\ capture, i.e., the case $\widetilde{x}%
_{t+1}=\widetilde{y}_{t}\ $ and $\widetilde{y}_{t+1}=\widetilde{x}_{t}$. But
this would mean
\[
d\left(  \widetilde{x}_{t},y_{t}\right)  =d\left(  \widetilde{x}%
_{t},\widetilde{y}_{t+1}\right)  =d\left(  \widetilde{x}_{t},\widetilde{x}%
_{t}\right)  =0;
\]
in other words, we would have capture in TBCR which contradicts the assumption
that $\sigma_{R}^{\ast}$ is a winning robber strategy. Hence \textquotedblleft
en passant\textquotedblright\ capture is also impossible in CCCR. The proof is complete.
\end{proof}

It is straightforward to extend the above for the case of $\widetilde
{c}\left(  G\right)  =K$ and obtain \ the following.

\begin{lemma}
\label{prp0306}$\widetilde{c}(G)=K\Rightarrow c(G)\leq K$.
\end{lemma}

Now we can prove our main result.

\begin{theorem}
\label{prp0307}$c\left(  G\right)  =K\Leftrightarrow\widetilde{c}\left(
G\right)  =K$.
\end{theorem}

\begin{proof}
Assume that $c\left(  G\right)  =K$. By Lemma \ref{prp0304} we have $c\left(
G\right)  =K\Rightarrow\widetilde{c}\left(  G\right)  \leq K$; if
$\widetilde{c}\left(  G\right)  =K^{\prime}<K$, then by Lemma \ref{prp0306} we
have $c(G)\leq K^{\prime}<K=c\left(  G\right)  $, which is a contradiction.
Thus $c\left(  G\right)  =K\Rightarrow\widetilde{c}\left(  G\right)  =K$.

Conversely, assume that $\widetilde{c}(G)=K$. By Lemma \ref{prp0306} we have
$\widetilde{c}(G)=K\Rightarrow c(G)\leq K$; if $c(G)=K^{\prime}<K$, then by
Lemma \ref{prp0304} we have $\widetilde{c}(G)=K^{\prime}<K=\widetilde
{c}\left(  G\right)  $, which is a contradiction. Thus $\widetilde
{c}(G)=K\Rightarrow c(G)=K$.
\end{proof}

\section{Time Optimality\label{sec04}}

\subsection{Existence of Value and Optimal Strategies}

Recall that $\Gamma_{\left(  x_{0},y_{0}\right)  }^{G}$ denotes the CCCR game
played on graph $G$ by a single cop starting at location $x_{0}$ and a single
robber starting at location $y_{0}$. We equip $\Gamma_{\left(  x_{0}%
,y_{0}\right)  }^{G}$ with a \emph{payoff function}, defined as follows. First
define the auxiliary function
\[
\forall\left(  x,y\right)  \in V^{2}:r\left(  x,y\right)  =\left\{
\begin{array}
[c]{ll}%
1 & \text{iff }x\neq y\\
0 & \text{iff }x=y
\end{array}
\right.
\]
where $x$ and $y$ are cop and robber locations, respectively. Suppose that for
every round of $\Gamma_{\left(  x_{0},y_{0}\right)  }^{G}$ in which the robber
remains uncaptured, C pays R one unit of utility and denote by $v_{\left(
x_{0},y_{0}\right)  }^{G}\left(  \widetilde{\pi}_{C},\widetilde{\pi}%
_{R}\right)  $ \ the total amount collected by R (obviously it depends on the
strategies $\widetilde{\pi}_{C},\widetilde{\pi}_{R}$). Then the \emph{payoff}
of $\Gamma_{\left(  x_{0},y_{0}\right)  }^{G}$ is
\[
v_{\left(  x_{0},y_{0}\right)  }^{G}\left(  \widetilde{\pi}_{C},\widetilde
{\pi}_{R}\right)  =\mathbb{E}\left(  \sum_{t=0}^{\infty}r\left(  x_{t}%
,y_{t}\right)  \right)
\]
where $\mathbb{E}\left(  \cdot\right)  $ denotes expected value and, for
notational brevity, the dependence of $x_{t},y_{t}$ on $\widetilde{\pi}%
_{C},\widetilde{\pi}_{R}$ has been suppressed.

Following the terminology of \cite{filar1997competitive}, we recognize that
CCCR equipped with the above payoff is a \emph{positive stochastic game}, R is
Player 1 or the \emph{Maximizer} and C is Player 2 or the \emph{Minimizer}.
These terms reflect the fact that R (resp. C)\ chooses $\widetilde{\pi}_{R}$
(resp. $\widetilde{\pi}_{C}$) to maximize (resp. to minimize)$\ v_{\left(
x,y\right)  }^{G}\left(  \widetilde{\pi}_{C},\widetilde{\pi}_{R}\right)  $. We
always have
\begin{equation}
\sup_{\widetilde{\pi}_{R}}\inf_{\widetilde{\pi}_{C}}v_{\left(  x,y\right)
}^{G}\left(  \widetilde{\pi}_{C},\widetilde{\pi}_{R}\right)  \leq
\inf_{\widetilde{\pi}_{C}}\sup_{\widetilde{\pi}_{R}}v_{\left(  x,y\right)
}^{G}\left(  \widetilde{\pi}_{C},\widetilde{\pi}_{R}\right)  .
\end{equation}
The following  is standard game theoretic terminology \cite{filar1997competitive}.
\begin{definition}
\label{prp0401}If we have%
\begin{equation}
\inf_{\widetilde{\pi}_{C}}\sup_{\widetilde{\pi}_{R}}v_{\left(  x,y\right)
}^{G}\left(  \widetilde{\pi}_{C},\widetilde{\pi}_{R}\right)  =\sup
_{\widetilde{\pi}_{R}}\inf_{\widetilde{\pi}_{C}}v_{\left(  x,y\right)  }%
^{G}\left(  \widetilde{\pi}_{C},\widetilde{\pi}_{R}\right)  \label{eq0501}%
\end{equation}
then we denote the common quantity of (\ref{eq0501}) by $\widehat{v}_{\left(
x,y\right)  }^{G}$ and call it the \emph{value} of $\Gamma_{\left(
x,y\right)  }^{G}$.
\end{definition}

\begin{definition}
\label{prp0401a}We denote the \emph{capture time of }$G$ by $CT\left(
G\right)  $ and define it by
\[
CT\left(  G\right)  =\max_{\left(  x,y\right)  \in V^{2}}\widehat{v}_{\left(
x,y\right)  }^{G}.
\]

\end{definition}

What is the connection between the $\Gamma_{\left(  x,y\right)  }^{G}$
for  various $\left(  x,y\right)  \in V^{2}$? It is natural to assume
that if at some stage of $\Gamma_{\left(  x,y\right)  }^{G}$  we reach the
position $\left(  x^{\prime},y^{\prime}\right)  $ then we can play the
remaining portion of $\Gamma_{\left(  x,y\right)  }^{G}$ as if \emph{we are
just  starting the game }$\Gamma_{\left(  x^{\prime},y^{\prime}\right)  }^{G}%
$. This plausible assumption can be proved rigorously (see
\cite{shapley1953stochastic} and \cite[pp.89-91]{filar1997competitive}) and
has the important consequence that, for a given $G$, $\widehat{v}_{\left(
x,y\right)  }^{G}$\emph{ is the same for every game }$\Gamma_{\left(
x^{\prime},y^{\prime}\right)  }^{G}$ (and hence it is correct to omit
mention of a specific game in the notations $v_{\left(  x,y\right)  }%
^{G}\left(  \widetilde{\pi}_{C},\widetilde{\pi}_{R}\right)  $ and
$\widehat{v}_{\left(  x,y\right)  }^{G}$). An additional important consequence
is the existence of \emph{memoryless optimal strategies which are the same for
all }$\Gamma_{\left(  x,y\right)  }^{G}$ \emph{games}, as will be seen in
Theorem \ref{prp0403}. Before stating and proving this theorem we need some
additional definitions.

\begin{definition}
\label{prp0402}Given $\varepsilon\geq0$, we say that the cop strategy
$\ \widetilde{\pi}_{C}^{\varepsilon}$ is $\varepsilon$\emph{-optimal} (for the
game $\Gamma_{\left(  x,y\right)  }^{G}$) iff%
\[
\left\vert \widehat{v}_{\left(  x,y\right)  }^{G}-\sup_{\widetilde{\pi}_{R}%
}v_{\left(  x,y\right)  }^{G}\left(  \widetilde{\pi}_{C}^{\varepsilon
},\widetilde{\pi}_{R}\right)  \right\vert \leq\varepsilon.
\]
Similarly, we say that the robber strategy $\widetilde{\pi}_{R}^{\varepsilon}$
is $\varepsilon$-optimal (for the game $\Gamma_{\left(  x,y\right)  }^{G}$)
iff
\[
\left\vert \widehat{v}_{\left(  x,y\right)  }^{G}-\inf_{\widetilde{\pi}_{C}%
}v_{\left(  x,y\right)  }^{G}\left(  \widetilde{\pi}_{C},\widetilde{\pi}%
_{R}^{\varepsilon}\right)  \right\vert \leq\varepsilon.
\]
A 0-optimal (cop or robber)\ strategy is simply called \emph{optimal}.
\end{definition}

If both $\widetilde{\pi}_{C}^{\ast}$ and $\widetilde{\pi}_{R}^{\ast}$ are
optimal, then we have $\widehat{v}_{\left(  x,y\right)  }^{G}=v_{\left(
x,y\right)  }^{G}\left(  \widetilde{\pi}_{C}^{\ast},\widetilde{\pi}_{R}^{\ast
}\right)  $. The main facts about the $\Gamma_{\left(  x,y\right)  }^{G}%
$\ games are summarized in the following.

\begin{theorem}
\label{prp0403}For every graph $G=\left(  V,E\right)  $ and every $\left(
x,y\right)  \in V^{2}$ the following hold.

\begin{enumerate}
\item For every $\left(  x,y\right)  \in V^{2}$, the game $\Gamma_{\left(
x,y\right)  }^{G}$\ has the value $\widehat{v}_{\left(  x,y\right)  }^{G}$. 

\item There exists a memoryless cop strategy $\widetilde{\pi}_{C}^{\ast}$
which is optimal  for every game $\Gamma_{\left(  x,y\right)  }^{G}$. For
every $\varepsilon>0$, there exists a memoryless robber strategy
$\widetilde{\pi}_{R}^{\varepsilon}$ which is $\varepsilon$-optimal for every
game $\Gamma_{\left(  x,y\right)  }^{G}$. 

\item $V^{2}$ can be partitioned into the sets $V_{1}$ and $V_{2}$ defined by%
\[
V_{1}=\left\{  \left(  x,y\right)  :\widehat{v}_{\left(  x,y\right)  }%
^{G}<\infty\right\}  ,\quad V_{2}=\left\{  \left(  x,y\right)  :\widehat
{v}_{\left(  x,y\right)  }^{G}=\infty\right\}  .
\]

\item If $c\left(  G\right)  =1$, then $V_{1}=V^{2}$, i.e., $\widehat
{v}_{\left(  x,y\right)  }^{G}<\infty$ for every $\left(  x,y\right)  \in
V^{2}$.
\end{enumerate}
\end{theorem}

\begin{proof}
Parts 1 and 2 follow immediately from the results of \cite{gurel2009pursuit}.
Part 3, the partition of $V^{2}$ into $V_{1}$ and $V_{2}$, is just a
definition. It remains to show part 4, i.e., that $c\left(  G\right)
=1\Rightarrow V_{1}=V^{2}$. This will also follow from \cite{gurel2009pursuit}
if we can show the existence of a cop strategy $\widetilde{\pi}_{C}^{\#}$ and
a constant $M_{G}$ such that
\begin{equation}
\forall\widetilde{\pi}_{R},x,y:v_{\left(  x,y\right)  }^{G}\left(
\widetilde{\pi}_{C}^{\#},\widetilde{\pi}_{R}\right)  \leq M_{G}<\infty
;\label{eq0511}%
\end{equation}
in other words, $\widetilde{\pi}_{C}^{\#}$ guarantees finite (not necessarily
optimal) capture time for every robber strategy and every starting position.

The required $\widetilde{\pi}_{C}^{\#}$ is the strategy used in the proof of
Lemma \ref{prp0303}. Indeed, recall that%
\[
A_{k}=\left\{  \mathbf{s}:\mathbf{s}\in\widetilde{S}^{\infty}\text{ and
}R\text{ is still free after the first }k\cdot T\text{ rounds}\right\}
\]
and when $C$ \ uses $\widetilde{\pi}_{C}^{\#}$ $\ $and $R$ uses any
$\widetilde{\pi}_{R}$ we have%
\[
\sum_{k=1}^{\infty}\Pr\left(  A_{k}\right)  \leq\sum_{k=1}^{\infty}\left(
1-\left(  \frac{1}{n}\right)  ^{T}\right)  ^{k}.
\]
It follows that
\begin{align*}
\forall\widetilde{\pi}_{R}  :v_{x_{0},y_{0}}^{G}\left(  \widetilde{\pi}_{C}^{\#},\widetilde{\pi}_{R}\right)  &  =\mathbb{E}\left(  \sum_{t=0}^{\infty}r\left(  x_{t},y_{t}\right)  \right) \\
                                                                                                             &  \leq\sum_{k=1}^{\infty}k\cdot T\cdot\Pr\left(  A_{k}\right)  \leq\sum_{k=1}^{\infty}k\cdot T\cdot\left(  1-\left(  \frac{1}{n}\right)  ^{T}\right)^{k}=\widetilde{M}_{G,x,y}<\infty.
\end{align*}
Letting $M_{G}=\max_{\left(  x,y\right)  \in V^{2}}\widetilde{M}%
_{G,x,y}<\infty$, where $M_{G}$ depends only on $G$, we see that
$\widetilde{\pi}_{C}^{\#}$ satisfies (\ref{eq0511}) and the proof is complete.
\end{proof}

\begin{remark}
\label{prp0404}\normalfont The theorem can be extended to the game
$\Gamma_{\left(  x,y\right)  }^{G,K}$\ for any graph $G$ (with any
$\widetilde{c}\left(  G\right)  $), any number of cops $K$ and any initial
position $\left(  x,y\right)  $ (we will now have $x\in V^{K}$). If
$K\geq\widetilde{c}\left(  G\right)  $, then $V_{1}=V^{K+1}$. Note that the
set $V_{1}$ will never be empty; for example, when $K=1$, $\left(  x,x\right)
$ belongs to $V_{1}$ for any $\widetilde{c}\left(  G\right)  \in\mathbb{N}$
(since $v_{\left(  x,x\right)  }^{G}\left(  \widetilde{\pi}_{C},\widetilde
{\pi}_{R}\right)  =0$ for any $G,x,\widetilde{\pi}_{C},\widetilde{\pi}_{R}$).
\end{remark}

\begin{remark}
\label{prp0405}\normalfont Parts 1, 2 and 3 of the theorem can also be proved
immediately using the results of either \cite{filar1997competitive} or
\cite{kumar1981existence}.
\end{remark}

\subsection{Computation of Value and Optimal Strategies}

The value and optimal strategies of $\Gamma_{\left(  x,y\right)  }^{G}$ can be
computed by \emph{value iteration}, as shown by Theorem \ref{prp0406}. Before
presenting the theorem and its proof let us give its intuitive justification.

Suppose at time $t$ the game position is $\left(  x,y\right)  $. As already mentioned,
we can assume that the \textquotedblleft\emph{remainder game}"
is $\Gamma_{\left(  x,y\right)  }^{G}$, i.e., it can be played as a new
CCCR\ game starting at $\left(  x,y\right)  $; the remainder game has value
$\widehat{v}_{\left(  x,y\right)  }^{G}$. Suppose further that $C$ uses the
move $u$ and $R$ uses the move $v$. The new game position is
\[
\left(  x^{\prime},y^{\prime}\right)  =\widetilde{Q}\left(  \left(  x^{\prime
},y^{\prime}\right)  ,u,v\right)  ,
\]
$R$ receives
\[
r\left(  x^{\prime},y^{\prime}\right)  =r\left(  \widetilde{Q}\left(  \left(
x,y\right)  ,u,v\right)  \right)
\]
units from $C$ and, invoking memorylessness again, the remainder-game is
$\Gamma_{\left(  x^{\prime},y^{\prime}\right)  }^{G}=\Gamma_{\widetilde
{Q}\left(  \left(  x,y\right)  ,u,v\right)  }^{G}$ and has value $\widehat
{v}_{\left(  x^{\prime},y^{\prime}\right)  }^{G}=\widehat{v}_{\widetilde
{Q}\left(  \left(  x,y\right)  ,u,v\right)  }^{G}$. To describe the
relationship between $\widehat{v}_{\left(  x,y\right)  }^{G}$ and $\widehat
{v}_{\widetilde{Q}\left(  \left(  x,y\right)  ,u,v\right)  }^{G}$ we need some
new notation.

Recall that a finite \emph{two-person zero-sum game in normal form} can be
specified by a single $M\times N$ matrix $P$ \cite{osborne1994game}. The game
is played in a single round as follows:\ \emph{simultaneously} the maximizing
Player 1 chooses the row index $m$ and the minimizing Player 2 chooses the
column index $n$; then Player 2 pays to Player 1 the amount $A_{mn}$. It is
well known that every such game has a value and many algorithms are available
to compute it. We denote the game matrix $A$ by the notation $\left\{
A_{mn}\right\}  _{m=1,...,M}^{n=1,...,N}$ and its value by $\mathbf{Val}%
\left[  \left\{  A_{mn}\right\}  _{m=1,...,M}^{n=1,...,N}\right]  $.

It seems reasonable (and can be rigorously justified) that $\Gamma_{\left(
x,y\right)  }^{G}$ can be considered as a single-round finite two-person
zero-sum game as follows: when $C$ chooses move $u$ and $R$ chooses move $v$
the payoff to $R$ is%
\begin{equation}
r\left(  \widetilde{Q}\left(  \left(  x,y\right)  ,u,v\right)  \right)
+\widehat{v}_{\widetilde{Q}\left(  \left(  x,y\right)  ,u,v\right)  }^{G}.
\label{eq0521}%
\end{equation}
In other words, $R$ receives $r\left(  \widetilde{Q}\left(  \left(
x,y\right)  ,u,v\right)  \right)  $ units as the payoff of the current round
and $\widehat{v}_{\widetilde{Q}\left(  \left(  x,y\right)  ,u,v\right)  }^{G}$
units as the payoff of the \textquotedblleft remainder-game" $\Gamma_{\widetilde{Q}\left(
\left(  x,y\right)  ,u,v\right)  }^{G}$ (which is assumed to be played
optimally by both players). Hence the game matrix of $\Gamma_{\left(
x,y\right)  }^{G}$ is $\left\{  r\left(  \widetilde{Q}\left(  \left(
x,y\right)  ,u,v\right)  \right)  +\widehat{v}_{\widetilde{Q}\left(  \left(
x,y\right)  ,u,v\right)  }^{G}\right\}  _{u\in V}^{v\in V}$ and has value%
\begin{equation}
\widehat{v}_{\left(  x,y\right)  }^{G}=\mathbf{Val}\left[  \left\{  r\left(
\widetilde{Q}\left(  \left(  x,y\right)  ,u,v\right)  \right)  +\widehat
{v}_{\widetilde{Q}\left(  \left(  x,y\right)  ,u,v\right)  }^{G}\right\}
_{u\in V}^{v\in V}\right]  . \label{eq0522}%
\end{equation}
Note that (\ref{eq0522})\ holds when $x\neq y$; for $x=y$ we obviously have
$\widehat{v}_{\left(  x,y\right)  }^{G}=0$.

The above is an informal argument for the connection between the values
$\widehat{v}_{\left(  x,y\right)  }^{G}$. The following theorem shows that the
argument can be made rigorous; furthermore, the theorem provides a method for
computing the values, as well as the optimal strategies.

\begin{theorem}
\label{prp0406}For every graph $G=\left(  V,E\right)  $ with $\widetilde
{c}\left(  G\right)  =1$ and for every $\left(  x,y\right)  \in V^{2}$ the
values $\left\{  \widehat{v}_{\left(  x,y\right)  }^{G}\right\}  _{\left(
x,y\right)  \in V^{2}}$ are the smallest (componentwise) positive solution of
the system of \emph{optimality equations}:%
\begin{align}
\widehat{v}_{\left(  x,y\right)  }^{G}  &  =\mathbf{Val}\left[  \left\{
r\left(  \widetilde{Q}\left(  \left(  x,y\right)  ,u,v\right)  \right)
+\widehat{v}_{\widetilde{Q}\left(  \left(  x,y\right)  ,u,v\right)  }%
^{G}\right\}  _{u\in V}^{v\in V}\right]  \quad\text{when }x\neq
y,\label{eq0541}\\
\widehat{v}_{\left(  x,y\right)  }^{G}  &  =0\quad\text{when }x=y.
\label{eq0542}%
\end{align}
Furthermore, for $n=0,1,2,...$, define the initial conditions
\[
v_{\left(  x,y\right)  }^{G}\left(  0\right)  \geq0\ \text{when }x\neq
y\quad\text{and}\quad v_{\left(  x,y\right)  }^{G}\left(  0\right)
=0,0\quad\text{when }x=y
\]
and, for $n\in\mathbb{N}$, the recursion (\emph{value iteration})
\begin{align}
v_{\left(  x,y\right)  }^{G}\left(  n+1\right)   &  =\mathbf{Val}\left[
\left\{  r\left(  \widetilde{Q}\left(  \left(  x,y\right)  ,u,v\right)
\right)  +v_{\widetilde{Q}\left(  \left(  x,y\right)  ,u,v\right)  }%
^{G}\left(  n\right)  \right\}  _{u\in V}^{v\in V}\right]  \quad\text{when
}x\neq y,\label{eq0543}\\
v_{\left(  x,y\right)  }^{G}\left(  n+1\right)   &  =0\quad\quad\text{when
}x=y. \label{eq0544}%
\end{align}
Then
\[
\forall\left(  x,y\right)  \in V^{2}:\lim_{n\rightarrow\infty}v_{\left(
x,y\right)  }^{G}\left(  n\right)  =\widehat{v}_{\left(  x,y\right)  }^{G}.
\]

\end{theorem}

\begin{proof}
This is essentially the combination of Theorems 4.4.3 and 4.4.4 from
\cite{filar1997competitive}, but the following modifications are required.

In \cite{filar1997competitive} the optimality equations (\ref{eq0541}%
)-(\ref{eq0542}) and the recursion (\ref{eq0543})-(\ref{eq0544}) are given in
terms of transition probabilities which in our notation would be written as
$P\left(  \left(  x^{\prime},y^{\prime}\right)  |\left(  x,y\right)
,u,v\right)  $; this is the probability that the new position is $\left(
x^{\prime},y^{\prime}\right)  $ given the old position is $\left(  x,y\right)
$ and the player moves are $u$ and $v$. However, in CCCR transitions are
deterministic, i.e.,
\[
P\left(  \left(  x^{\prime},y^{\prime}\right)  |\left(  x,y\right)
,u,v\right)  =\left\{
\begin{array}
[c]{ll}%
1 & \text{when }\left(  x^{\prime},y^{\prime}\right)  =\widetilde{Q}\left(
\left(  x,y\right)  ,u,v\right) \\
0 & \text{otherwise.}%
\end{array}
\right.
\]
Furthermore, once the game reaches a capture position $\left(  x,x\right)  $,
it will always stay in this position, which has value $\widehat{v}_{\left(
x,x\right)  }^{G}=0$.

Taking the above in account, the optimality equations and the recursion of
\cite{filar1997competitive} reduce to (\ref{eq0541})-(\ref{eq0544}).
\end{proof}

\begin{remark}
\label{prp0407}\normalfont The modification of the theorem theorem for the
game $\Gamma_{\left(  x,y\right)  }^{G,K}$, with $K>1$, is obvious.
\end{remark}

We conclude this section with some examples. We apply the value iteration
(\ref{eq0543})-(\ref{eq0544}) to several graphs and discuss the results.

\begin{example}
\label{prp0408}\normalfont In the first example $G$ is a path of five nodes,
as illustrated in Figure \ref{fig01}. 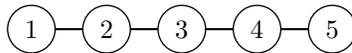
\begin{figure}[H]
\begin{center}
\begin{tikzpicture}
\SetGraphUnit{2}
\Vertex[x=5+1,y=0]{1}
\Vertex[x=5+2,y=0]{2}
\Vertex[x=5+3,y=0]{3}
\Vertex[x=5+4,y=0]{4}
\Vertex[x=5+5,y=0]{5}
\Edge(1)(2)
\Edge(2)(3)
\Edge(3)(4)
\Edge(4)(5)
\SetVertexNoLabel
\end{tikzpicture}
\end{center}
\par
\label{fig01}\caption{A five node path.}%
\end{figure}\noindent Let the $\left(  x,y\right)  $ element of matrix
$\widehat{V}^{G}$ be equal to $\widehat{v}_{\left(  x,y\right)  }^{G}$, the
value of $\Gamma_{\left(  x,y\right)  }^{G}$ when the cop is at node $x$ and
the robber at node $y$. Value iteration yields%
\[
\widehat{V}^{G}=\left[
\begin{array}
[c]{lllll}%
0 & 4 & 4 & 4 & 4\\
1 & 0 & 3 & 3 & 3\\
2 & 2 & 0 & 2 & 2\\
3 & 3 & 3 & 0 & 1\\
4 & 4 & 4 & 4 & 0
\end{array}
\right]  .
\]
Cop and robber optimal strategies can be described quite easily:\ the cop
should always move towards the robber and the robber should always move away
from the robber\footnote{Actually the robber has several other optimal
strategies; he can also stay in place if the cop is at a distance greater than
one. These are pure (i.e., deterministic) strategies; the robber also has
mixed optimal strategies.}. Clearly in this graph the CCCR time-optimal
strategies are the same as those for the TBCR\ game.
\end{example}

\begin{example}
\label{prp0409}\normalfont Not surprisingly, for the tree $G$ illustrated in
Figure \ref{fig02}, the optimal cop and robber strategies are again the same
for the CCCR\ and TBCR\ games.

\begin{figure}[H]
\begin{center}
\begin{tikzpicture}
\SetGraphUnit{2}
\Vertex[x= 0,y= 0]{1}
\Vertex[x=-1,y=-1]{2}
\Vertex[x= 1,y=-1]{3}
\Vertex[x= 0,y=-2]{4}
\Vertex[x= 2,y=-2]{5}
\Edge(1)(2)
\Edge(1)(3)
\Edge(3)(4)
\Edge(3)(5)
\SetVertexNoLabel
\end{tikzpicture}
\end{center}
\par
\label{fig02}\caption{A tree.}%
\end{figure}\noindent The value-iteration algorithm yields%
\[
\widehat{V}^{G}=\left[
\begin{array}
[c]{lllll}%
0 & 1 & 2 & 2 & 2\\
3 & 0 & 3 & 3 & 3\\
2 & 2 & 0 & 1 & 1\\
3 & 3 & 3 & 0 & 2\\
3 & 3 & 3 & 2 & 0
\end{array}
\right]
\]

\end{example}

\begin{example}
\label{prp0410}\normalfont
The next example involves the clique of three nodes, as illustrated in Figure
\ref{fig03}.

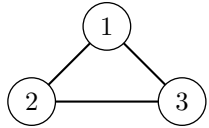
\begin{figure}[H]
\begin{center}
\begin{tikzpicture}
\SetGraphUnit{2}
\Vertex[x= 0,y= 0]{1}
\Vertex[x=-1,y=-1]{2}
\Vertex[x= 1,y=-1]{3}
\Edge(1)(2)
\Edge(1)(3)
\Edge(2)(3)
\SetVertexNoLabel
\end{tikzpicture}
\end{center}
\par
\label{fig03}\caption{A three node clique.}%
\end{figure}\noindent After eight iterations, the algorithm yields
$\overline{V}^{G}$ which is (componentwise) within 10$^{-2}$ of the true
solution%
\[
\widehat{V}^{G}=\left[
\begin{array}
[c]{lll}%
0 & 2 & 2\\
2 & 0 & 2\\
2 & 2 & 0
\end{array}
\right]  .
\]
The algorithm also yields the optimal strategies, which are symmetrical with
respect to the cop and robber positions $\left(  x_{0},y_{0}\right)  $. For
example, when $\left(  x_{0},y_{0}\right)  =\left(  3,1\right)  $ we have%
\[
\pi_{C}^{\ast}=\left[  \frac{1}{2},\frac{1}{2},0\right]  \text{ and }\pi
_{C}^{\ast}=\left[  0,\frac{1}{2},\frac{1}{2}\right]  .
\]
In other words, under these strategies, both cop and robber always move with
equal probability to one of the two nodes they don't currently occupy. It can
be verified analytically that these strategies yield the previously displayed
value matrix $\widehat{V}^{G}$. Because of symmetry, many other optimal
strategies exist for both cop and robber.
\end{example}

\begin{example}
\label{prp0411}\normalfont The final example involves a Gavenciak graph
\cite{gavenciak2010copwin}, as illustrated in Figure \ref{fig04}.

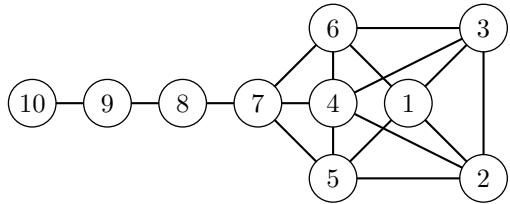
\begin{figure}[H]
\begin{center}
\begin{tikzpicture}
\SetGraphUnit{2}
\Vertex[x= 2,y=-1]{1}
\Vertex[x= 3,y=-2]{2}
\Vertex[x= 3,y= 0]{3}
\Vertex[x= 1,y=-1]{4}
\Vertex[x= 1,y=-2]{5}
\Vertex[x= 1,y= 0]{6}
\Vertex[x= 0,y=-1]{7}
\Vertex[x=-1,y=-1]{8}
\Vertex[x=-2,y=-1]{9}
\Vertex[x=-3,y=-1]{10}
\Edge(1)(2)
\Edge(1)(3)
\Edge(1)(5)
\Edge(1)(6)
\Edge(2)(3)
\Edge(2)(5)
\Edge(2)(4)
\Edge(3)(6)
\Edge(3)(4)
\Edge(4)(5)
\Edge(4)(6)
\Edge(4)(7)
\Edge(5)(7)
\Edge(6)(7)
\Edge(7)(8)
\Edge(8)(9)
\Edge(9)(10)
\SetVertexNoLabel
\end{tikzpicture}
\end{center}
\par
\label{fig04}\caption{A Gavenciak graph.}%
\end{figure}\noindent From the results of \cite{gavenciak2010copwin} we know
that the TBCR\ capture time of this graph is 7 (this is the minimax of the
capture time over all initial positions). The cop is able to achieve this
capture time by first maneuvering himself to node 7 and forcing the robber into the
path subgraph, and then chasing the robber all the way to node 10. In the CCCR
game the results are similar but they require the use of randomized
strategies. We do not present the entire $\widehat{V}^{G}$ (because of space
limitations) but let us give some indicative results. For example, when the
initial positions are $\left(  x_{0},y_{0}\right)  =\left(  2,1\right)  $, the
cop cannot be certain of capturing the robber in one move (since the moves are
simultaneous). It turns out that, by the application of randomized strategies,
the optimal expected capture time is $\widehat{v}_{2,1}^{G}\cong18.82...$ .
However, the part of the strategies which concerns the path subgraph is, as in
TBCR, deterministic. For inctance once the cop reaches node 8 (with the robber
in either node 9 or 10) he should deterministically perform the transitions 
$8\rightarrow 9\rightarrow10$. Let us also note that for this ten-nodes
graph, the value iteration algorithm required 90 iterations to get
(componentwise) within 10$^{-2}$ of the true solution.
\end{example}

\section{Related Work\label{sec05}}

While the assumption of simultaneous moves is a natural one (and is better
than turn-based movement as a model of real world pursuit / evasion problems)
it appears that CCCR\ has not been studied in the cops and robber literature.
However, our analysis of time optimal CCCR strategies follows closely the
corresponding study of time optioimal TBCR\ strategies presented in
\cite{hahn2006game} (and expanded in \cite{bonato2012general}). Both Hahn's
algorithm in \cite{hahn2006game} and the recursion (\ref{eq0543}) of Theorem
\ref{prp0406} are value iteration algorithms. The main difference between the
two is this:\ while in Hahn's algorithm updating the value in every iteration
only requires taking a minimum or a maximum, \emph{every value iteration } 
of (\ref{eq0543}) requires solving a one-round,
zero-sum game (this is indicated by the $\mathbf{Val}\left[  \cdot\right]  $ operator
in (\ref{eq0543})). Consequently, (\ref{eq0543})\ is computationally more
intensive than Hahn's algorithm.

As already mentioned, simultaneous moves have not been explored in the
CR\ literature. On the other hand, an interesting analog can be found in the
literature of \emph{reachability games}
\cite{berwangergraph,mazala2002infinite}. As we have pointed out in
\cite{kehagias2014role,kehagias2014zermelo}, TBCR is a special case of a
\textquotedblleft classical\textquotedblright\ (i.e., turn-based) reachability
game. Similarly, CCCR\ is a special case of a \emph{concurrent} reachability
game\footnote{This is the source of our term \textquotedblleft concurrent CR game\textquotedblright.}. The
literature on concurrent reachability games \cite{alfaro2007concurrent} can
furnish useful insights for the analysis of CCCR.

All the above problems can be considered as special cases of the general
\emph{stochastic game}. The book \cite{filar1997competitive} is an excellent,
comprehensive and relatively recent study of the topic; it also contains many
references to important earlier work.

\section{Concluding Remarks\label{sec06}}

We conclude this paper by presenting questions which, in our opinion, merit
further study.

One group of questions concern the definition of cop number, which in turn
depends on the definition of the properties of the \emph{capture event}. To
understand the issue, we must turn back to concurrent reachability games. Let
$A$ be the set of all histories of a reachability game and $B\subseteq A$ the
set of all realizations in which the target state is reached (in CCCR, $B$
would be the set of all infinite histories $\left\{  \left(  x_{t},y_{t}\right)  \right\}  _{t=0}^{\infty}$ 
for which there is some $t_{c}$ such
that $x_{t_{c}}=y_{t_{c}}$). As pointed out in \cite{alfaro2007concurrent},
the target state can be reached in at least three different senses (and each
of these implies the next one in the list).

\begin{enumerate}
\item \emph{Sure reachability}:\ $B=A$.

\item \emph{Almost sure reachability}:\ $\Pr\left(  B\right)  =1$.

\item \emph{Limit sure reachability}: For every real $\varepsilon$, player 1 has a
strategy such that for all strategies of player 2, the target state is reached
with probability greater than $1-\varepsilon$.
\end{enumerate}

\noindent The above carry over to CCCR and can be used to define corresponding
cop numbers: $c_{sure}\left(  G\right)  $, $c_{almostsure}\left(  G\right)  $
and $c_{limitsure}\left(  G\right)  $. In this paper we have worked
exclusively with $\widetilde{c}\left(  G\right)  =c_{almostsure}\left(
G\right)  $. Obviously
\[
c_{sure}\left(  G\right)  \geq c_{almostsure}\left(  G\right)  \geq
c_{limitsure}\left(  G\right)
\]
but several additional questions can be asked. For example can the ratios
$\frac{c_{sure}\left(  G\right)  }{c_{almostsure}\left(  G\right)  }$ and
$\frac{c_{sure}\left(  G\right)  }{c_{limit-sure}\left(  G\right)  }$ be
bounded by a constant?\ By a number depending on the size of $G$? How about
the differences $c_{sure}\left(  G\right)  -c_{almostsure}\left(  G\right)  $
and $c_{almostsure}\left(  G\right)  -c_{limitsure}\left(  G\right)  $?

Another group of questions concerns the CCCR\ variants obtained by modifying
the cop's and/or the robber's behavior.

\begin{enumerate}
\item For example, what is the \emph{cost of drunkenness}? In other words,
what is the ratio of expected capture times for the previously descibed CCCR
game and a variant in which the robber performs a random walk on the nodes of
the graph?\ The same question has been studied for the TBCR\ case in
\cite{kehagias20102some,kehagias2013cops}

\item Similarly, what is the \emph{cost of visibility}? In this case we study
the ratio of expected capture times for the previously descibed CCCR game and
a variant in which the robber is invisible to the cop. For the TBCR\ case,
this has been studied in \cite{kehagias2013cops,kehagias2014role}
\end{enumerate}

\end{document}